\documentclass{article}
\usepackage[utf8]{inputenc}
\usepackage{indentfirst}
\usepackage{graphicx,tikz, tikz-cd}
\usetikzlibrary{decorations.pathreplacing}
\usetikzlibrary{arrows}
\usepackage{fancyhdr}
\usepackage{hyperref}
\hypersetup{pdftitle={Entanglement},pdfauthor={Lorenzo Panebianco}}
\graphicspath{ {images/} }
\usepackage[T1]{fontenc}
\usepackage{mathrsfs}
\usepackage[english]{babel}
\usepackage{float,floatflt, epsfig}
\usepackage{subfig}
\usepackage{color}
\usepackage{caption}
\usepackage{upgreek}
\usepackage{fancyhdr}
\usepackage{lipsum} 
\usepackage{latexsym}
\usepackage{mathtools}
\usepackage{mathtools}

\usepackage{amssymb,amsthm,amsmath}
\usepackage{bbm}
\usepackage{amscd,xypic}
\usepackage{mathrsfs}
\usepackage{braket}
\usepackage[toc,page]{appendix}
\usepackage{enumerate} 
\usepackage{booktabs}
\usepackage{mathtools}
\usepackage[all]{xy}
\usepackage[margin=1.1in]{geometry}

\newcommand{\vertiii}[1]{{\left\vert\kern-0.25ex\left\vert\kern-0.25ex\left\vert #1 
		\right\vert\kern-0.25ex\right\vert\kern-0.25ex\right\vert}}

\newcommand{\C}{\mathbb{C}}

\newcommand{\R}{\mathbb{R}}

\newcommand{\hil}{\mathcal H}

\newcommand{\hi}{\mathcal{H}}
\newcommand{\al}{\mathcal{A}}

\newcommand{\M}{\mathcal{M}}

\theoremstyle{plain}
\newtheorem{thm}{Theorem}

\newtheorem*{thm*}{Theorem}
\newtheorem{cor}[thm]{Corollary}
\newtheorem*{cor*}{Corollary}

\newtheorem{lem}[thm]{Lemma} 
\newtheorem*{lem2a*}{Lemma 2A} 
\newtheorem*{lem2b*}{Lemma 2B} 
\newtheorem*{lem2c*}{Lemma 2C} 
\newtheorem{prop}[thm]{Proposition}

\newtheorem*{claim*}{Claim} 

\theoremstyle{definition}

\newtheorem{defn}[thm]{Definition}

\newtheorem*{defn*}{Definition}

\theoremstyle{remark}
\newtheorem{remark}[thm]{Remark}

\title{Modular Nuclearity and  Entanglement Measures}

\author{{Lorenzo Panebianco}${}^{1}$\footnote{{\tt panebianco@mat.uniroma1.it} , ${}^\triangle$ {\tt wegener@mat.uniroma2.it} }, { Benedikt Wegener}${}^{2,3,\triangle}$.
}
\date{\small{
		${}^{1}$Dipartimento di Matematica,  Università di Roma “La Sapienza” \\
		Piazzale Aldo Moro 5, 00185–Roma, Italy \\
		${}^{2}$ Marie Sklodowska-Curie fellow of the Istituto Nazionale di Alta Matematica\\
		${}^{3}$ Dipartimento di Matematica, Universit\`a di Roma “Tor Vergata” \\
		Via della Ricerca Scientifica 1, I-00133 Roma, Italy}\\
}

\begin{document}

	\maketitle

	\begin{abstract}
		
		In the framework of Algebraic Quantum Field Theory,  several operator algebraic notions of entanglement entropy can be associated to any couple of causally disjoint and distant spacetime regions $\mathcal{S}_A$  and $\mathcal{S}_B$. One of them, known as canonical entanglement entropy, is defined as the von Neumann entropy on some canonical intermediate type I factor. In this work  we show the canonical entanglement entropy of the vacuum state to be finite on a wide family of conformal nets including the $U(1)$-current model and the $SU(n)$-loop group models.  More in general, such a finiteness property is expected to  rely on some nuclearity condition of the system.  To support this conjecture, we show that the mutual information is finite in any local QFT verifying a  modular $p$-nuclearity condition for some $0 < p <1$.  A similar result is proved for another entanglement entropy introduced in this work.  We conclude with some personal considerations on $1+1$-dimensional integrable models with factorizing S-matrices and remarks for future works in this direction. 
	\end{abstract}
	
	
	
	
	\section{Introduction} \label{sec:1}

	In classical information theory, a standard notion is that of Shannon entropy, a quantity that measures  the amount of information carried by a given state of a system. Its quantum mechanical analogue was defined by and named after von Neumann. In this context, the von Neumann entropy is connected to the quantum phenomenon of entanglement. Entanglement occurs on a composite system when a state is not the product of two independent states on the subsystems. Entanglement has been investigated profoundly as a means of probing the foundations of quantum mechanics (as in the EPR paradox and Bell's inequalities) as well as a resource for quantum information theory. On type I factors, the  von Neumann entropy of a state $\varphi$ with density matrix $\rho$ is
	\[
	S(\varphi)=-\text{tr}\, \rho\log \rho \,.
	\]

	However, in the axiomatic approach to Quantum Field Theory (QFT) utilizing Haag-Kastler nets one can show under very general conditions that the algebras of observables are type III von Neumann algebras \cite{fredenhagen1985modular}, where no density matrices exist and a different entropy-type functional is required. The role of entanglement in  QFT is more recent and increasingly important \cite{longo2021neumann}. It appears in relations with several primary research topics in theoretical physics as area theorems \cite{hollands2018entanglement}, c-theorems \cite{casini2007c} and quantum null energy inequalities \cite{morinelli2021modular, panebianco2021loop}. \\ 
	
	In the study of entropy in local QFT, nuclearity conditions play an important role \cite{hollands2018entanglement,narnhofer1994entropy,otani2018toward}. These conditions are predicated on the compactness criterion proposed by Haag and Swieca \cite{haag1965does}. This criterion is an abstraction of the insight that a QFT model must have a bound on the number of its local  degrees of freedom to show regular thermodynamical behaviour. Based on similar heuristic arguments, Buchholz and Wichmann strengthened this assumption and suggested the first nuclearity condition \cite{buchholz1986causal}, which nowadays is known as \emph{energy nuclearity condition}. Explicitly, let  $\mathcal{O} \mapsto \mathcal{A}(\mathcal{O}) \subseteq B(\hil)$ be some local Haag-Kastler net describing a local QFT on the Minkowski space. Denote by $\Omega$ the vacuum vector and by $\omega$  the corresponding vacuum state. One says that the energy nuclearity condition holds if  
	\begin{equation} \label{eq:enc0}
		\Theta_{\beta, \mathcal{O}} \colon \mathcal{A}(\mathcal{O}) \to \hil \,, \quad 	\Theta_{\beta, \mathcal{O}} (a) = e^{- \beta H} a \Omega \,,
	\end{equation}
	is nuclear for any bounded region $\mathcal{O}$ and any inverse temperature $\beta >0$, where $H = P_0$ denotes the Hamiltonian with respect to the time direction $x_0$.  The nuclear norm of  \eqref{eq:enc0} can be interpreted as the partition function of the restricted system at some fixed inverse temperature, hence it is natural to expect  such nuclear norms to measure  the entropy of the state. At a later time, a related nuclearity condition has been found by use of  modular theory \cite{buchholz1990nuclear, buchholz1990nuclear2}. According to this second nuclearity condition, one considers an inclusion $\mathcal{O} \subset \widetilde{\mathcal{O}}$ of spacetime regions and requires the map
	\begin{equation}\label{eq:mnc0}
		\Xi \colon \mathcal{A}(\mathcal{O}) \to \hil \,, \quad \Xi(x) = \Delta^{1/4} x \Omega \,,
	\end{equation}
	to be nuclear, with $\Delta$ the modular operator of the bigger local algebra $\mathcal{A}(  \widetilde{\mathcal{O}})$ with respect to $\Omega$. In this case one speaks of {\em modular nuclearity condition}. If the map \eqref{eq:mnc0} is $p$-nuclear then one will say that the {\em modular $p$-nuclearity condition} is satisfied. If modular $p$-nuclearity holds for some $ 0 < p \leq 1$ then the modular nuclearity condition is satisfied, and if so then it is well known that the {\em split property} holds, namely there is an intermediate type I factor $ \mathcal{A}(  {\mathcal{O}})  \subset \mathcal{F} \subset \mathcal{A}(  \widetilde{\mathcal{O}})$ \cite{buchholz1990nuclear, buchholz1990nuclear2}.  \\
	
	The split property is an expression of the statistical independence of two spacelike separated regions. It is considerably stronger than (Einstein) causality, which is included in the Haag-Kastler axioms. While causality expresses the independence of measurement apparatuses located in spacelike separated regions, the split property additionally takes the stage of preparation into account. It signifies that no choice of a state prepared in one region can prevent the preparation of any state in the other region. Among all intermediate type I factors, a canonical one can be chosen by using standard representation arguments \cite{doplicher1984standard}. When the split property holds, one natural entanglement measure to consider is the von Neumann entropy on the canonical intermediate type I-factor  \cite{longo2021neumann}, namely what we here call {\em canonical entanglement entropy}. As $\omega$ is pure on $B(\hil) = \mathcal{F} \vee \mathcal{F}'$, the canonical entanglement entropy is given by  
	\[
	E_C(\omega)  = S_\mathcal{F} (\omega) = S_{\mathcal{F}'}(\omega) \,.
	\]
	In this work we extend some results of \cite{longo2021neumann} and we show this entanglement entropy to be finite on a wide family of conformal nets including the $U(1)$-current. In order to do so, we first notice where the proof of \cite{longo2021neumann} can be extended and we then construct a vacuum preserving conditional expectation between canonical intermediate type I factors. \\
	
	More in general, it is a common belief that this  entanglement entropy can be bounded from above  by assuming  some nuclearity property of the system. A result of this type is  strongly suggested by  previous works \cite{longo2021neumann, narnhofer1994entropy, otani2018toward} and shows interesting applications in AdS/CFT contexts as  pointed out in \cite{dutta2021canonical}, where it appeared  as splitting entropy and it was used to conjecture a bound on the reflected entropy. This type of considerations can be extended to a wide family of entanglement measures defined on a generic  bipartite system $A \otimes B$ like the  {\em mutual information} \cite{hollands2018entanglement}
	\[
	E_I(\omega) = S(\omega \Vert \omega_A \otimes \omega_{B}) \,,
	\]
	which quantifies how much knowledge on the second system one can gain from measuring the first one. In this work we prove a relation between the mutual information and the modular $p$-nuclearity condition. \\

	The paper is organised as follows.  In \autoref{sec:2}, \autoref{sec:3} and  \autoref{sec:4}  we collect various notions about entropy, entanglement and modular nuclearity.  In \autoref{sec:5} we apply the theory of standard representations to construct a cpu map between canonical intermediate type I factors, with a focus on twisted-local nets on the circle. We then use this construction to extend \cite{longo2021neumann}. In \autoref{sec:6} we study the connection between modular $p$-nuclearity conditions and certain entanglement measures. In particular, we show that the modular $p$-nuclearity with $0 < p < 1$ implies the finiteness of the mutual information and of one more entanglement measure inspired by \cite{otani2018toward}. We also add remarks concerning area laws  \cite{hollands2018entanglement}. Finally, in \autoref{sec:7} we apply these considerations to a family of  $1+1$-dimensional integrable models with factorizing S-matrices \cite{lechner2008construction} and investigate the asymptotic behaviour of different  entanglement measures as the distance between two causally disjoint wedges diverges. In \autoref{sec:8} we present a few remarks that might be useful for future research in this direction.

	\section{Quantum entropy basics} \label{sec:2}
	
	Let $\mathcal{M}$ be a von Neumann algebra in standard form on some Hilbert space $\hi$ and let $\varphi$, $\psi$ be two normal positive linear functionals on $\mathcal{M}$ represented by the  vectors $\xi$ and $ \eta$. Denote by $s(\varphi)=[\M' \xi]$ and $s(\psi) = [\M' \eta]$ the central supports of $\varphi$ and $\psi$ respectively.   We define the {\em Tomita relative operator}
	\[
	S_{\xi, \eta}(x \eta + \zeta) = s(\psi)x^* \xi \,, \quad x \in \M \,, \; \zeta \in [\M \eta]^\perp \,.
	\]
	This densely defined conjugate-linear operator is closable. Its closure will be equally denoted and its polar decomposition is given by
	\[
	S_{\xi, \eta} = J_{\xi, \eta} \Delta_{\xi, \eta}^{1/2} \,,
	\]
	where
	\[
	\text{supp} \, \Delta_{\xi, \eta} = s(\varphi)s'(\psi) \,, \quad J^*_{\xi, \eta}J_{\xi, \eta}= s(\varphi)s'(\psi)\,, \quad J_{\xi, \eta}J^*_{\xi, \eta}= s'(\varphi)s(\psi)\,.
	\]
	In the case $\xi = \eta$  we will write $S_\xi = S_{\xi, \xi}$, and similarly $J_{\xi} = J_{\xi, \xi}$ and $\Delta_{\xi} = \Delta_{\xi, \xi}$. If $\xi$ and $\eta$ are both in the natural cone $\mathcal{P}^\natural$, then we also have the polar decomposition
	\[
	S_{\xi, \eta} = J\Delta_{\xi, \eta}^{1/2} \,,
	\] 
	with $J$ the modular conjugation of the standard form.  We recall that if $\varphi$ is faithful then $\xi$ is cyclic, and if so then $\mathcal{P}^\natural = \mathcal{P}^\natural_\xi$ where   
	\begin{equation} \label{eq:natural}
	\mathcal{P}^\natural_\xi =  \overline{ \big\{ \Delta^{1/4}_\xi  x^*x \xi \,, \, x \in \mathcal{M} \big\} } =  \overline{ \big\{ xJx^*J \xi \,, \, x \in \mathcal{M} \big\} }  \,.
	\end{equation}
Finally, by using primes to denote the modular operators of the commutant, we have the identities
	\[
	J_{\xi, \eta} \Delta_{\xi, \eta}^{1/2}  J_{\xi, \eta} = \Delta_{\eta, \xi}^{-1/2}  \,, \quad J'_{\xi, \eta} = J_{\eta, \xi} \,, \quad (\Delta_{\eta, \xi}')^z = \Delta_{ \xi, \eta}^{-z} \,.
	\]
	The {\em relative entropy} between $\varphi$ and $\psi$ is defined by
	\begin{equation} \label{eq:re}
		S(\varphi \Vert \psi)  = -(\xi|\log \Delta_{\eta, \xi} \xi) 
	\end{equation}
	if  $s(\varphi) \leq s(\psi)$, otherwise $S(\varphi \Vert \psi) = + \infty$ by definition. Equation \eqref{eq:re} does not depend on the choice of the representing vectors.  If $\M$ is not in standard form, then equation \eqref{eq:re} holds if the relative modular operator is replaced with a spatial derivative \cite{ohya2004quantum}. \\
	
	The  scalar product \eqref{eq:re} has to be intended by applying the spectral theorem to the relative modular operator $\Delta_{\eta, \xi}$, namely we have
	\begin{equation} \label{eq:re2}
		S(\varphi \Vert \psi)  = - \int_0^1 \log \lambda \, d(\xi | E_{\eta, \xi}(\lambda) \xi ) - \int_1^\infty  \log \lambda \, d(\xi | E_{\eta, \xi}(\lambda) \xi ) \,,
	\end{equation}	
	where the second integral is always finite by the estimate $\log \lambda  \leq \lambda$. In particular, $S(\varphi \Vert \psi)  $ is finite if and only if the first integral appearing in \eqref{eq:re2} is finite. By this remark it follows that \cite{ohya2004quantum}
	\begin{equation} \label{eq:re3}
		S(\varphi \Vert \psi) = i \frac{d}{dt}\varphi((D \psi \colon D \varphi)_t)  \bigg\vert_{t=0} = - i \frac{d}{dt}\varphi((D \varphi \colon D \psi )_t )   \bigg\vert_{t=0}  \,,
	\end{equation}
	where $(D \varphi \colon D \psi )_t = (D \psi \colon D \varphi)_t^*$ is the Connes cocycle. Identity \eqref{eq:re3} can be proved by using the dominated convergence theorem if $S(\varphi \Vert \psi) $ is finite and by the Fatou's lemma if  $S(\varphi \Vert \psi) = + \infty$ as shown in \cite{ciolli2020information}. Let us now recall some properties of the relative entropy \cite{ohya2004quantum}.

	\begin{itemize}
		\item[(i)] $S(\varphi \Vert \psi)  \geq \varphi(I)  (\log \varphi(I) - \log \psi(I))$,  and $S(\lambda \varphi \Vert \mu \psi) = \lambda S(\varphi \Vert \psi) - \lambda \varphi(I)\log(\mu / \lambda)$ for any $\lambda, \mu \geq 0$. Moreover, $S(\varphi \Vert \psi) \geq \Vert \varphi - \psi \Vert^2/2$, so that  $S(\varphi \Vert \psi)  = 0$ if and only if $\varphi = \psi$. 
		
		\item[(ii)] $S(\varphi \Vert \psi) $ is lower semi-continuous in the $\sigma(\M_*, \M)$-topology. 
		
		\item[(iii)] $S(\varphi \Vert \psi) $ is convex in both its variables. By (i) this is equivalent to the subadditivity of $S(\varphi \Vert \psi) $ in both its variables.
		
		\item[(iv)] $S(\varphi \Vert \psi) $ is superadditive in its first argument. Furthermore, $S(\varphi \Vert \psi) \leq S(\varphi' \Vert \psi')$ if $\psi \geq \psi'$ and $\varphi \geq \varphi'$ with $\Vert  \varphi \Vert = \Vert \varphi' \Vert$.
		
		\item[(v)] If $\alpha \colon \M_1 \to \M_2$ is a Schwarz mapping such that $\varphi_2 \cdot \alpha \leq \varphi_1$ and $\psi_2 \cdot \alpha \leq \psi_1$, then $S_{\M_1}(\varphi_1 \Vert \psi_1) \leq S_{\M_2}(\varphi_2 \Vert \psi_2)$. In particular, $S(\varphi \Vert \psi)$ is monotone increasing with respect to inclusions of von Neumann algebras. 
		
		\item[(vi)] Let  $(\M_i)_{i }$ be  an increasing net of von Neumann subalgebras of $\M$ with the property  $(\cup_i \M_i)'' = \M$. Then the increasing net $ S_{\M_i}(\varphi \Vert \psi)  $ converges to $S(\varphi \Vert \psi) $. 
		
		\item[(vii)] Let $\varepsilon \colon \M \to \mathcal{N}$ be a faithful normal  conditional expectation. If $\varphi$ and $\psi$ are normal states on $\M$ and $\mathcal{N}$ respectively, then $S_{\M}(\varphi \Vert \psi \cdot \varepsilon) = S_{\mathcal{N}}(\varphi \Vert \psi)  + S_{\M}(\varphi \Vert \varphi \cdot \varepsilon) $. 
		
		\item[(viii)] Let $\varphi$ be a normal state on the spatial tensor product  $\M_1 \otimes \M_2$ with partials $\varphi_i = \varphi \vert_{\M_i}$. Consider then  normal states $\psi_i$  on $\M_i$. As a corollary of (vi), we have $S(\varphi \Vert \psi_1 \otimes \psi_2) = S(\varphi_1 \Vert \psi_1 ) + S(\varphi_2  \Vert \psi_2 ) + S(\varphi \Vert \varphi_1 \otimes \varphi_2)$. 
	\end{itemize} 
	
	By using the universal representation, the relative entropy can be defined on a generic  $C^*$-algebra. If  we replace the strong closure with the norm closure in (r5) and the $\sigma(\M_*, \M)$-topology with the weak topology in (r1),  then properties from (r0) to (r5) still hold in the $C^*$-algebraic setting \cite{ohya2004quantum}. As shown in \cite{ohya2004quantum}, if $\psi$ is a positive normal functional of $\M$ and  $t \in \R$, then the sublevel 
	\begin{equation} \label{eq:sublevel}
		\mathcal{K}(\psi, t) = \{ \varphi \in \M^*_+ \colon S( \varphi \Vert \psi ) \leq t \}
	\end{equation}
	consists of normal functionals and it is a convex compact set with respect to the $\sigma(\M_*, \M)$-topology. The following lemma is original.
	
	\begin{lem}
		$\mathcal{K}(\psi, t) $ is sequentially $\sigma(\M_*, \M)$-compact and its set of extremal points is  
		\begin{equation} \label{eq:ep}
			\mathcal{E}(\psi, t) = \{ \varphi \in \mathcal{M}^*_+ \colon S( \varphi \Vert \psi) = t \}\,.
		\end{equation}
		Moreover, after a restriction to $\mathcal{M}_{s(\psi)}$, the union  $\mathcal{K}(\psi) = \bigcup_{t} \mathcal{K}(\psi, t)  $  is  norm dense in the set of  normal positive functionals of the reduced algebra  $\mathcal{M}_{s(\psi)}$.
	\end{lem}
	\begin{proof}
		The first two claims follow from the Eberlein-Smulian theorem and  Donald's identity (\cite{ohya2004quantum}, Proposition 5.23). The last point holds if $\psi$ is faithful, since in this case the set of positive normal functionals $\varphi$ such that $\varphi  \leq \alpha \psi$ for some $\alpha >0$ is norm dense in $\mathcal{M}_*^+$ (\cite{bratteli2012operator1}, Theorem 2.3.19). The general case follows by noticing that  $S_{\mathcal{M}}(\varphi \Vert \psi) = S_{\mathcal{M}_{s(\psi)}}(\varphi \Vert \psi)$ if $s(\varphi) \leq s(\psi)$, which is a necessary condition for $S_{\mathcal{M}}(\varphi \Vert \psi)$ to be finite. 
	\end{proof}

	
	\begin{defn}
		If $\varphi$ is a state on a $C^*$-algebra $A$, then the {\em von Neumann entropy} of $\varphi$ is defined by
		\[
		S_A(\varphi) = \sup\Big\{  \sum_i \lambda_i S(\varphi_i \Vert \varphi) \colon \sum_i \lambda_i \varphi_i  = \varphi \Big\} \,,
		\]
		where the supremum is over all decompositions of $\varphi$ into finite (or equivalently countable) convex combinations of other states. If $A$ is clear, we will simply write $S_A(\varphi) = S(\varphi)$. 
	\end{defn}
	Some properties of $S(\varphi)$ are immediate from those of the relative entropy: $S(\varphi)$ is nonnegative, vanishes if and only if $\varphi$ is a pure state and it is weakly lower semicontinuous. On type I factors, the von Neumann entropy of a normal state $\varphi$ with density matrix $\rho$ is given by  $S(\varphi) = - \text{tr} \, \rho \log \rho$. We now list a few properties of the von Neumann entropy \cite{ohya2004quantum}. The notation $\eta(t) = -t \log t$ is standard in information theory. 
	
	\begin{itemize}
		\item[(s0)] (concavity) Given states $\varphi$ and $\omega$, then $\lambda S(\varphi) + (1-\lambda)S(\omega) \leq S(\lambda \varphi + (1-\lambda)\omega) \leq \lambda S(\varphi) + (1-\lambda) S(\omega)+ \eta(\lambda) + \eta(1-\lambda)$ for each $0 < \lambda < 1$.
		\item[(s1)] (strong subadditivity) On a three-fold-product $B(\hil_1) \otimes B(\hil_2) \otimes B(\hil_3)$, a normal state $\omega_{123}$ with marginal states $\omega_{ij}$ satisfies $S(\omega_{123}) + S(\omega_2) \leq S(\omega_{12}) + S(\omega_{23})$. 
		\item[(s2)]  $S(\psi) = \inf \big\{-\sum_i \eta(\lambda_i) \big\}$,  where the infimum is taken over all the possible decompositions into pure states. 
		\item[(s3)] (tensor product) On the projective tensor product $A \otimes B$, we have the identity  $S(\varphi_1 \otimes \varphi_2) = S(\varphi_1) + S(\varphi_2)$. 
	\end{itemize}

	\begin{defn}
		Consider an inclusion of $C^*$-algebras $ A \subseteq B$ and a state $\varphi$ on $B$. The {\em entropy of $\varphi$  with respect to $A$} is 
		\begin{equation} \label{eq:ce}
			H_\varphi^B(A) = \sup \Big\{ \sum_i \lambda_i S_A(\varphi_i  \Vert \varphi )  \colon \varphi = \sum_i \lambda_i \varphi_i \Big\} \,,
		\end{equation}
		where the supremum is over all finite (countable) convex decompositions $\varphi = \sum_i \lambda_i \varphi_i$ on $B$. If the bigger $C^*$-algebra $B$ is clear, we will briefly use the notation $H_\varphi^B(A) = H_\varphi(A) $.
	\end{defn}
	The entropy of a subalgebra is actually a particular case of what is known as {\em conditional entropy}. We list a few of its properties \cite{connes1987dynamical, ohya2004quantum}. 
	\begin{itemize}
		\item[(c0)] (monotonicity) $H_\varphi^{{B_2} }(A_1) \leq H_\varphi^{B_1}({A_2} )$ if $A_1 \subseteq {A_2} \subseteq B_1 \subseteq {B_2} $.
		
		\item[(c1)] (semicontinuity)  $\varphi \mapsto H_\varphi^B(A) $ is weakly lower semicontinuous.
		
		\item[(c2)] (martingale property) $ \lim_i H_\varphi^B(A_i) = H_\varphi^B(A) $ if $(A_i)_i$ is an increasing net of $C^*$-subalgebras of $B$ with union norm dense in $A$.
		
		\item[(c3)] (concavity) $\lambda H^B_{\varphi_1} (A) + (1 - \lambda) H^B_{\varphi_2} (A) \leq H^B_{\varphi} (A) \leq  \lambda H^B_{\varphi_1} (A) + (1 - \lambda) H^B_{\varphi_2} (A)  + \eta(\lambda) + \eta(1-\lambda)  $ for $\varphi = \lambda \varphi_1 + (1- \lambda) \varphi_2$ on $B$ and $\lambda$ in $(0,1)$. 
	\end{itemize}

	In (c2), the union $\cup_i A_i$ can be strongly dense if all the $C^*$-algebras are replaced with von Neumann algebras and the state $\varphi$ is normal. We point out that the concavity of $H_\varphi(A)$ mentioned in (c3) certainly holds whenever $A$ is AF (\cite{ohya2004quantum}, Theorem 5.29 and Proposition 10.6), but the general case is a bit unclear to the authors \cite{connes1987dynamical}. What is clear  instead,	is the  following   original simple lemma which says whenever the  inequality (c0) reduces to  an equality in the case $A_1 = {A_2}$.

	\begin{lem} \label{lem:inequality}
		Consider the $C^*$-algebras inclusions $A_1 \subseteq B_1 \subseteq {B_2}$ and $A_1 \subseteq {A_2} \subseteq {B_2}$. Let  $\varphi$ be a state on ${B_2}$. If there is a $\varphi$-preserving conditional expectation $\varepsilon \colon {B_2}  \to {B_1} $, then 
		\[
		H_\varphi^{B_1}(A_1) \leq H_\varphi^{{B_2}}({A_2}) \,.
		\]
	\end{lem}
	
	\begin{proof}
		We can follow Proposition 6.7 of \cite{ohya2004quantum}. Indeed, if  $\psi$ is a state of $B_1$ then $\psi \cdot \varepsilon $ is a state of ${B_2}$. Therefore, if  $\varphi = \sum_i \lambda_i \varphi_i $ on $B_1$ for some states  $\varphi_i$ of  $B_1$ then $\varphi = \sum_i \lambda_i \varphi_i \cdot \varepsilon $  is a decomposition of $\varphi$ into states of ${B_2}$. The rest follows from $S_{A_1}(\varphi_i \Vert \varphi) \leq S_{{A_2}}(\varphi_i \cdot \varepsilon  \Vert \varphi) $. 
	\end{proof}

	In general, the existence of some $\varphi$-preserving conditional expectation between von Neumann algebras is not always true. A well known necessary and sufficient condition has been provided by Takesaki (\cite{takesaki2013theory}, Theorem IX.4.2). However, if it exists then it satisfies the following structure theorem, 	which is a partial generalization of  Proposition 3.1.4 and Proposition 3.1.5 of \cite{jones1983index}  to the non-tracial case. 
	
	\begin{thm} \label{thm:jones}
		Let $\mathcal{M}$ be a von Neumann algebra in standard form on  $\hil$. Let $\mathcal{N} $ be a von Neumann subalgebra of $\M$, $\varphi$ a normal faithful state of $\M$ and  $\xi$  the unique vector in the natural cone representing $\varphi$. If there exists a $\varphi$-preserving conditional expectation $\varepsilon \colon \M\to \mathcal{N}$, then there is a unique  projection $e$ in $\mathcal{N}'$ such that

		(i)  $ex \xi = \varepsilon(x)\xi$ and $exe = \varepsilon(x)e$ for $x$ in $\M$. 
		
		Furthermore, 
		
		(ii) $\mathcal{N}e = e (\mathcal{M} \vee e)e$, 
		
		(iii) $\mathcal{N}' = \mathcal{M}' \vee e$,

		(iv) there is a $*$-isomorphism $\phi \colon \mathcal{N}e \to \mathcal{N}$ such that $\varepsilon (x)= \phi (exe)$ for $x$ in $\M$. 
	\end{thm}
	\begin{proof}
(i) Recalling that $\xi$ is standard for $\mathcal{M}$, one can first show  the map $x \xi \mapsto \varepsilon(x) \xi$ to be an orthogonal projection onto $[\mathcal{N} \xi]$ since $\varepsilon$ is a $\varphi$-preserving conditional expectation. It is then easy to notice that $e$ belongs to $\mathcal{N}'$, and the identity $exe = \varepsilon(x)e$  follows. The uniqueness follows by construction. (ii) This relation is a corollary of the first point.  (iii) This follows from   $\mathcal{N}e = e (\mathcal{M} \vee e)e$, which is a consequence of (i) as well. (iv) We show the map $y  \mapsto ye$ for $y$ in $\mathcal{N}$ to be injective and hence an isomorphism. But if $ye=0$ then $ey^*=0$, hence $\varphi(y^*y)=0$ which implies $y=0$ by faithfulness. 
	\end{proof}
	
	\begin{remark}
If we denote by $s(\varphi)$ the central support of $\varphi$ and we assume $\mathcal{N}$ to be a von Neumann subalgebra of $s(\varphi) \M s(\varphi)$, then the previous theorem still holds in the nonfaithful case. 
	\end{remark}
	
	\begin{cor} \label{cor:natural}
		Let $\M$, $\mathcal{N}$, $\varphi$, $\xi$ and $\varepsilon$ be as in the previous theorem. If $\mathcal{P}^\natural_\M $ and $\mathcal{P}^\natural_{\mathcal{N}} $ are the natural cones of $\M$ and $\mathcal{N}$ respectively, then
		\[
		e \, \mathcal{P}^\natural_\M  = \mathcal{P}^\natural_{\mathcal{N}} \subseteq \mathcal{P}^\natural_\M \,.
		\] 
		 In particular, the elements of  $\mathcal{P}^\natural_{\mathcal{N}} $ correspond to the $\varepsilon$-invariant normal positive functionals of $\M$.
	\end{cor}
\begin{proof}
	Let $S = J \Delta^{1/2}$ be the modular operator of $\M$ with respect to $\xi$. By the previous theorem we have $\varepsilon(x^*) = \varepsilon(x)^*$ and hence $eS = Se$, which implies that $e \Delta = \Delta e$ is the modular operator of $\mathcal{N}e$. Similarly, $eJ=Je$ is the modular conjugation of $\mathcal{N}e$. The thesis follows from the first identity of equation \eqref{eq:natural}, where the last claim is a consequence of the uniqueness of the representative vector.
\end{proof}


\section{Entanglement Measures} \label{sec:3}

In this section we discuss entanglement in a general setting and we review some quantitative  measures of entanglement and their properties \cite{hollands2018entanglement}. \\

Let $A$, $B$ be a couple of commuting von Neumann algebras on some Hilbert space $\hil$. We shall say that the pair $(A, B)$ is {\em split} if there exists a von Neumann algebra isomorphism $\phi \colon A \vee B \to A \otimes B$ such that $\phi(ab) = a \otimes b$. We will refer to split pairs also as {\em bipartite systems}.  In quantum field theory, bipartite systems are associated to causally disjoint regions.  Clearly the spatial tensor product $A \otimes B$ has a natural structure of bipartite system since $ A \cong A \otimes \mathbbm{1}$ and  $B \cong \mathbbm{1} \otimes B$.  If $ A \vee B$ is $\sigma$-finite, then the pair $(A, B) $ is split if and only if for any given normal states $\varphi_A$ on $A$ and $\varphi_B$ on $B$ there exists a normal state $\varphi$ on $A \vee B$ such that  $\varphi (ab) = \varphi_A(a) \varphi_B(b)$ \cite{longo2020lectures}. \\

 A state $\omega$ on $A \otimes B$ is said to be {\em separable} if there are positive normal functionals $\varphi_j$ on $A$ and $\psi_j$ on $B$ such that $\omega = \sum_j \varphi_j \otimes \psi_j$, where the sum is assumed to be norm convergent. Separable states are normal and satisfy the following original lemma inspired by \cite{narnhofer2002entanglement}. 
 
 \begin{lem} \label{lem:entropy}
 	Given two von Neumann algebras $A$ and $B$, consider a state  $\omega$  on the bipartite system $A \otimes B$. If $\omega$ is separable, then
 	\[
 	S_A(\omega) = H_\omega^{A \otimes B}(A) \,.
 	\]
 \end{lem}
 
 \begin{proof}  If $\omega = \sum_j \phi_j \otimes \psi_j$ and $\pi$ is the GNS representation of $A$ associated to the marginal state $\omega_A = \omega \vert_A$, then we have an equivalence of GNS representations $\pi \cong  \oplus_j \pi_j$, where $\pi_j$  is the GNS representation of $A$ given by $\psi_j(1) \phi_j$. We can define a cpu map $\varepsilon_j \colon A \otimes B \to \pi_j (A)$ by $\varepsilon_j(a \otimes b) = \pi_j(a) \psi_j(b)/\psi_j(1)$, and this lead us to define a conditional expectation $\varepsilon \colon A \otimes B \to pAp$ by $\varepsilon = \pi^{-1} \cdot \oplus_j \varepsilon_j$, where $p$ is the support projection of $\omega_A$. The claim follows from the identity $S_A(\omega) = S_{pAp}(\omega)  $ (see the proof of Proposition 6.8 of \cite{ohya2004quantum}) and Lemma \ref{lem:inequality}.
 \end{proof}

 A normal state which is not separable is called {\em entangled}.	Therefore, an entanglement measure for a bipartite system should be a state functional that vanishes on separable states. 



\begin{defn}
	The {\em relative entanglement entropy} of a normal state $\omega$ on a bipartite system $A \otimes B$ is given by
	\begin{equation} \label{eq:ree}
		E_R(\omega) = \inf \big\{ S(\omega \Vert \sigma)  \colon \sigma \text{ is a separable state}\big\}\,.
	\end{equation}
	The {\em mutual information} $E_I(\omega)$ is given by 
	\[
	E_I(\omega) = S(\omega \Vert \omega_A \otimes \omega_B) \,.
	\]
	where $\omega_A= \omega \vert_A$ and similarly for $B$. 
\end{defn}

Clearly $E_R(\omega) \leq E_I (\omega)$. As an example, let us consider  a bipartite system given by  $A=B(\hil)$ and $B=B(\hil')$, with $\hil$ and $\hil'$ finite dimensional Hilbert spaces. The mutual information is given by
\begin{equation} \label{eq:mi}
	E_I(\omega) = S(\omega_A) + S(\omega_B) - S (\omega) \,.
\end{equation}
We point out that, without any finiteness assumption, on hyperfinite type I factors we can only write $	E_I(\omega) + S (\omega)  = S(\omega_A) + S(\omega_B) $. The mutual information is non-negative and independent of the order of $A$ and $B$ \cite{hollands2018entanglement}. For separable states $\omega = \sum_j \lambda_j  \varphi_j \otimes \psi_j$, with $\varphi_j$ and $\psi_j$ normal states, we also have \cite{hollands2018entanglement}
\begin{equation} \label{eq:e1}
	E_I(\omega) \leq \sum_j \eta(\lambda_j) \,,
\end{equation}
with $\eta(t) = -t \ln t$ the information function. Moreover, we can use \eqref{eq:mi} to deduce the following “concavity” property of the mutual information.

\begin{lem}\label{lem:concave-mutual-info}
	Let $A$ and $B$ be AF factors and $\omega=\sum_j \lambda_j \omega_j$ be a convex decomposition of a normal state $\omega$ on $A\otimes B$, with $\omega_j $ normal states. Then 
	\begin{align*}
		\sum_j \lambda_j E_I(\omega_j) - \sum_j \eta(\lambda_j) \leq E_I (\omega) \leq \sum_j \lambda_j E_I(\omega_j) +2 \sum_j \eta(\lambda_j).
	\end{align*}
\end{lem}
\begin{proof}
	This result is known for type $I$-factors, e.g. \cite{shirokov2017tight}. If $A$ and $B$ are generic hyperfinite factors, then there is a family of bipartite systems of finite-dimensional (and hence type $I$) factors that is weakly dense in $A\otimes B$. On each of these finite-dimensional systems, the statement holds by the previous argument. We conclude the statement using the approximation property (vi) of the relative entropy.
\end{proof}

 It is an easy remark to notice that, always by assuming $A$ and $B$ to be finite dimensional type I factors, if $\omega = \sum_j \lambda_j \omega_j$ is a convex  decomposition of a state $\omega$ in states $\omega_j$, then
\begin{center}
	$\sum_j \lambda_j E_I(\omega_j) - \sum_j \eta(\lambda_j) \leq E_I (\omega) \leq \sum_j \lambda_j E_I(\omega_j) +2 \sum_j \eta(\lambda_j) \,.$
\end{center}
By monotonicity of the relative entropy, the same inequalities hold if $\omega$ is normal and $A$ and $B$ are both hyperfinite type I  factors. Furthermore, if $\omega$ is pure then $E_I(\omega) = 2S(\omega_A) = 2S(\omega_B)$ (Proposition 6.5. of \cite{ohya2004quantum}) while the relative entanglement entropy between $A$ and $B$ is \cite{vedral1998entanglement}
\[
E_R(\omega) = S(\omega_A) = S(\omega_B) \,.
\]

\begin{defn} \label{def:sep}
	A cp map $\mathcal{F} \colon A_1 \otimes B_1 \to A_2 \otimes B_2$ between two bipartite systems will be called {\em local} if it is of the form
	\[
	\mathcal{F}(a \otimes b) = \mathcal{F}_A(a )  \otimes \mathcal{F}_B( b) \,,
	\]
	where $\mathcal{F}_A$ and $\mathcal{F}_B$ are normal cp maps. More generally, a {\em separable operation} is by definition a family of normal, local cp maps $\mathcal{F}_j$ such that $\sum_j \mathcal{F}_j(1) = 1$. We think of such an operation as mapping a state $\omega$ with probability $p_j = \omega(\mathcal{F}_j(1))$ to $\mathcal{F}^*_j \omega / p_j$.
\end{defn}
Separable operations map separable states to separable states.  The relative entanglement entropy \eqref{eq:ree} of a bipartite system $A \otimes B$  has the following properties \cite{hollands2018entanglement}.

\begin{itemize}
	\item[(e0)] (symmetry) $E_R(\omega)$ is independent of the order of the systems $A$ and $B$.\footnote{More precisely, we should say that $E_R(\omega) = E_R(\omega \cdot \pi)$, with $ E_R(\omega \cdot \pi)$ the relative entanglement entropy on $B \otimes A$ and $\pi$ the natural permutation isomorphism.}
	\item[(e1)] (non-negative) $E_R(\omega) \in [0, \infty]$, with $E_R(\omega)=0$  if $\omega$ is separable and $E_R(\omega)= \infty$ when $\omega$ is not normal. Furthermore, if $E_R(\omega)=0$ then $\omega$ is norm limit of separable states.
	\item[(e2)] (continuity) Let $(\mathfrak{A}_i
	)_i$ and $(\mathfrak{B}_i
	)_i$ be two  increasing nets of subalgebras of $A$ and $B$ respectively, with $\mathfrak{A}_i \cong \mathfrak{B}_i  \cong M_{n_i}(\C)$. Let $\omega_i$ and $\omega_i'$ be normal states on $\mathfrak{A}_i \otimes  \mathfrak{B}_i $ such that $\lim_i \Vert \omega_i - \omega_i' \Vert = 0$. Then 
	\[
	\lim_{i \to \infty} \frac{E_R(\omega_i') - E_R(\omega_i)}{\ln n_i} =0 \,.
	\]
	\item[(e3)] (convexity) $E_R$ is convex. 
	\item[(e4)] (monotonicity under separable definitions) Consider a separable operation described by cp maps $\mathcal{F}_j $ with $\sum_j \mathcal{F}_j (1) =1$. Then
	\[
	\sum_j p_j E_R(\mathcal{F}^*_j \omega / p_j) \leq E_R(\omega) \,,
	\]
	where the sum is over all $j$ with $p_j = \omega(\mathcal{F}_j(1))>0$. 
	\item[(e5)] (tensor products) Let ${A}_i \otimes {B}_i$ with $i=1,2$ be two bipartite systems, and let $\omega_i$ be states on ${A}_i \otimes {B}_i $. Then
	\[
	E_R(\omega_1 \otimes \omega_2) \leq E_R(\omega_1 ) + E_R( \omega_2) \,.
	\]
\end{itemize}

The mutual information \eqref{eq:mi} clearly satisfies (e0) and (e5), and it is shown in \cite{hollands2018entanglement} that it also  satisfies properties (e2) and (e4). Property (e1) does not follow in a straightforward way from the definitions, since we can state inequality \eqref{eq:e1} at most. Property  (e3) does not hold in general.

.

	\section{Modular nuclearity conditions} \label{sec:4}

We will say that a split pair $(A, B)$ is {\em standard} if $A$, $B$ and $A \vee B$ are in standard form with respect to some vector $\Omega$. We  set $J_{ A} = J_{A, \Omega}$, $J_B = J_{B, \Omega}$, and similarly $\Delta_A = \Delta_{ A , \Omega}$, $\Delta_B = \Delta_{ B , \Omega}$. As $A \otimes B$ is in standard form with respect to  $\Omega \otimes \Omega$,  the isomorphism $\phi \colon A \vee B \to A \otimes B$ has a {\em standard implementation}, namely is uniquely implemented by some unitary $U$ which maps the natural cone of $A \vee B$ onto the natural cone of $A \otimes B$  \cite{doplicher1984standard}. It can also be shown that $J_{ A} \otimes J_{B} = U J U^{-1} $, with $J = J_{ A \vee B, \Omega}$. The {\em canonical intermediate type I factors} are $F= U^{-1}(B(\hil) \otimes \mathbbm{1})U$ and $F' = U^{-1}( \mathbbm{1} \otimes B(\hil))U$. By construction, $F $ is the unique $J$-invariant type I factor $A \subseteq {F} \subseteq B'$, and similarly for $F'$. If $A$ and $B$ are both factors then ${F} = A \vee J A J = B' \cap J B' J$, and therefore ${F'} = B \vee J B J = A' \cap J A' J$ \cite{doplicher1984standard}.  \\ 

 An inclusion   $N \subseteq M$  of von Neumann algebras  is said to be {\em split} if the pair $(N, M')$ is split.  We shall often pass from a split inclusion to a
 split pair and back. The trivial inclusion $N=M$ is split if and only if $N$ is a type I factor 	\cite{longo2020lectures}. The inclusion  $N \subseteq M$ is said to be {\em standard} if there is a vector $\Omega$ which is standard for $N$, $M$ and the relative commutant $N' \cap M$. If $N \subseteq M$ is a standard split inclusion then each intermediate type I factor $R$ is $\sigma$-finite and hence separable, therefore the Hilbert space $\hil$ has to be separable as $R \Omega$ is dense in $\hil$. If  $N \vee M'$ has a cyclic and separating vector, then the pair $(N, M')$ is split if and only if  there is an intermediate type I factor  $N  \subseteq R \subseteq M$ \cite{longo2020lectures}.

\begin{defn}
	Consider an inclusion $N \subseteq M$ of von Neumann algebras on a Hilbert space $\hil$. Assume the existence of a standard vector $\Omega$ for $M$ and denote by $\Delta$ the corresponding modular operator. We will say that the inclusion $N \subseteq M$  satisfies the {\em modular nuclearity condition} if the map
	\begin{equation} \label{eq:xi}
		\Xi \colon N \to \hil \,, \quad \Xi(x) = \Delta^{1/4} x \Omega \,,
	\end{equation}
	is nuclear. 
\end{defn}

A modular nuclear  inclusion of factors is split, and a split  inclusion of factors implies the compactness of the map \eqref{eq:xi} \cite{buchholz1990nuclear}. This motivates the interest in the split property in local quantum field theory contexts, where the split property amounts to some form of statistical independence between causally disjoint spacetime regions \cite{hollands2018entanglement, lechner2008construction}. \\

The previous nuclearity condition can be easily generalized as follows. Consider a linear map $\Theta \colon \mathcal{E} \to \mathcal{F}$ between Banach spaces. The map $\Theta$ is said to be of {\em type $l^p$}, $p>0$, if there exists a sequence of linear mappings $\Theta_i \colon \mathcal{E} \to \mathcal{F}$ of rank $i$ such that   \cite{buchholz1990nuclear2}
\begin{center}
	$	\sum_i \Vert  \Theta - \Theta_i \Vert^p < + \infty \,.$
\end{center}
The map $\Theta $ will be said to be of {\em type $s$} if it is of type $l^p$ for any $p>0$. Each mapping $\Theta$ of type $l^p$ for some $0 < p \leq 1$ is nuclear. Indeed, there are sequences of linear functionals $e_i \in \mathcal{E}^*$ and of elements $f_i$ in $\mathcal{F}$ such that 
\begin{center}\label{eq:pnc1}
	$		\Theta(x) = \sum_i e_i(x)f_i \,, \quad x \in \mathcal{E} \,, $ 
\end{center}
is an absolutely convergent series for each $x$ in $\mathcal{E}$, with
\begin{center}\label{eq:pnc2}
	$		\Theta(x) = \sum_i e_i(x)f_i \,, \quad \Theta(x) = \sum_i \Vert e_i \Vert^p \Vert f_i \Vert^p < + \infty \,. $ 
\end{center}
The induced quasi-norm, also called {\em $p$-norm}, is given by
\begin{center}
	$  \Vert \Theta \Vert_p = \inf  \Big( \sum_i \Vert e_i \Vert^p \Vert f_i \Vert^p \Big)^{1/p}  \,,$
\end{center}
where the infimum is taken over all  possible representations of $\Theta$ of the form \eqref{eq:pnc1}. The above nuclearity condition  can be then rephrased as {\em modular $p$-nuclearity condition} if the map \eqref{eq:xi} is of type $l^p$. 

\begin{defn}
	Let $(A, B)$ be a standard split pair with standard vector  $\Omega$ representing a  state $\omega$. Denote by $\Delta_A$ and $\Delta_B$ the corresponding modular operators. We define 
	\begin{equation} \label{eq:sn}
	\Xi_A(a) = \Delta_{B'}^{1/4} a \Omega  	\,, \quad \Xi_B(b) = \Delta_{A'}^{1/4} b \Omega   \,.
	\end{equation}
	with $a$ in $A$ and $b$ in $B$. Given $p>0$ we define the {\em $p$-partition function} as 	
	\begin{equation} \label{eq:nu}
		z_p = \min \{ \Vert \Xi_A \Vert_p , \Vert \Xi_B \Vert_p \}  \,.
	\end{equation}
	We will say that the pair $(A, B)$ satisfies the {\em $p$-modular nuclearity condition} if the $p$-partition function is finite. 	In the case $p=1$ we will simply talk of {\em partition function} and  of {\em modular nuclearity condition}.
\end{defn}

The  $p$-modular nuclearity condition implies the modular nuclearity condition if $p \leq 1$. In order to motivate our definition, we notice that if $(A, B)$ satisfies the  modular nuclearity condition then it is split.

\section{Results in chiral CFT} \label{sec:5}

We begin with a few definitions. Let $\mathcal{K}$ be the family of all the open, nonempty and non dense intervals of the circle.  For $I$ in $\mathcal{K}$, $I'$ denotes the interior of the complement. The M\"obius group $\text{M\"ob} $ acts on the circle by linear fractional transformations. A {\em M\"obius covariant net} $(\mathcal{A}, U, \Omega)$ consists of a family $\{\mathcal{A}(I) \}_{I \in \mathcal{K}}$ of von Neumann algebras acting on a separable Hilbert space $\hil$, a strongly continuous unitary representation $U$ of M\"ob and a vector $\Omega$ in $\hil$, called the {\em vacuum vector}, satisfying the following properties  \cite{d2001conformal}:
\newline 

(i) $\al(I_1) \subseteq \al(I_2)$ if $I_1 \subseteq I_2$ (isotony),

(ii) $U(g) \al(I) U(g)^* = \al(g.I)$ for every $g $ in M\"ob and $I$ in $\mathcal{K}$ (M\"obius covariance),

(iii)  the representation $U$ has {\em positive energy}, namely the generator of rotations has non-negative spectrum (positivity of the energy),

(iv) $\Omega$ is cyclic for the von Neumann algebra $ \bigvee_{I \in \mathcal{K}} \al(I)$, and up to a scalar $\Omega$ is the  unique M\"ob-invariant vector of $\hil$ (vacuum). \\

A M\"obius covariant net is said to be {\em twisted-local} if the following axiom is satisfied:  \\

(v) there exists a unitary $Z$ commuting with the representation $U$ and such that $Z\Omega = \Omega $ and  $Z \al(I')Z^* \subseteq \al(I)'$   (twisted-locality).\\

In the following, twisted-local M\"obius covariant nets will be briefly referred to as twisted-local nets. A few consequences of the axioms (i)-(iv) are  \cite{d2001conformal} \\

(vi) $\Omega$ is cyclic and separating for each $\al(I) $  (Reeh Schlieder property),

(vii) $\al(I) \subseteq \bigvee_{\alpha} \al(I_\alpha)$ if $I \subseteq \bigcup_\alpha I_\alpha$ (additivity), \\

while if we also assume (v) then we have   \cite{d2001conformal} \\

(viii) $\al(I') = Z\al(I)'Z^*$ for every $I$ in $\mathcal{K}$ (twisted-duality),

(ix) if $I_+$ is the upper half of the circle and $\Delta$ is the modular operator associated to $\al(I_+)$ and $\Omega$, then for every $t$ in $\R$ we have 
\begin{equation} \label{eq:BW}
	\Delta^{it} = U (\delta_{-2 \pi t}) \,,
\end{equation}
where $\delta$ is the one parameter group of dilations (Bisognano-Wichmann),

(x) each local algebra $\al(I)$ is a type III factor and $\bigvee_{I \in \mathcal{I}_\R} \al(I) = B(\hil)$, with $\mathcal{I}_\R$ the set of all the open, nonempty and non dense intervals of $ S^1 \setminus \{-1 \}$ (irreducibility). \\

A twisted-local net is said to be {\em local} if $Z$ is the identity. A particular class of local nets is certainly that of {\em conformal nets} \cite{panebianco2021loop}. In a conformal net the following additional property is automatic \cite{morinelli2018conformal}: \\

(xi) if $\bar{I} \subset J$ then there is a type I factor $\mathcal{R}$ such that $\al(I) \subset \mathcal{R} \subset \al(J)$ (split property). \\

As noticed above, every intermediate type I factor $\mathcal{R}$ is separable, hence the split property ensures the separability of the Hilbert space on which the local algebras $ \al(I)$ act. \\

Let $(\mathcal{A}, U, \Omega)$ be a twisted-local net on a Hilbert space $\hil$. We call a family $\mathcal{B} = \{  \mathcal{B}(I) \}_{I \in \mathcal{K}}$ of von Neumann subalgebras $\mathcal{B}(I) \subseteq \mathcal{A}(I) $ a {\em subnet of $\mathcal{A}$} if  it satisfies isotony and M\"obius covariance with respect to $U$. We will use the notation $\mathcal{B} \subseteq \mathcal{A}$ to denote the subnets  $\mathcal{B} $ of $ \mathcal{A}$. If $\mathcal{A}(I)' \cap \mathcal{B}(I) = \C $ for one (and hence for all) interval $I$ in $\mathcal{K}$, then the inclusion $\mathcal{B} \subseteq \mathcal{A}$ is said to be {\em irreducible}. If  we  denote by $e= [\hil_\mathcal{B} ]$ the orthogonal projection onto  $\hil_\mathcal{B} = \overline{\bigvee_{I \in \mathcal{K}} \mathcal{B}(I) \Omega}$, then it is easy to notice that $e$ is in the commutant of all the von Neumann algebras $\mathcal{B}(I)$ and that it commutes with $U$. Then $\mathcal{B}$ is itself a M\"obius covariant net on $e \hil$ with unitary representation the restriction of $U$ on $e \hil$. The projection $e$ does not depend on the choice of the interval $I$ and is the Jones projection of  the inclusion $\mathcal{B}(I) \subseteq \mathcal{A}(I)$. Therefore, how noticed in Theorem \ref{thm:jones},   $\mathcal{B}(I)$  is naturally isomorphic to the reduced von Neumann algebra $e\mathcal{B}(I)e$. If we have an inclusion of nets $\mathcal{B} \subseteq \mathcal{A}$, then for each interval $I$ there is a canonical faithful normal conditional expectation $\varepsilon_I \colon \mathcal{A}(I) \to \mathcal{B}(I)$ which preserves the vacuum state $\omega$ by Bisognano-Wichmann, M\"obius covariance for the subnet $\mathcal{B}$ and   Takesaki's theorem (\cite{takesaki2013theory}, Theorem IX.4.2.).

\begin{prop} \label{prop:conditional}
	Let $\mathcal{B} \subseteq \mathcal{A}$ be an inclusion  of  M\"obius covariant  nets. If $\mathcal{A}$ is twisted-local, then all the conditional expectations $\varepsilon_I \colon \mathcal{A}(I) \to \mathcal{B}(I)$ extend to a unique vacuum-preserving conditional expectation $\varepsilon \colon \mathfrak{A} \to \mathfrak{B}$. Here  the $C^*$-algebra $\mathfrak{A} $ is the  norm closure of the union of the local algebras  $\mathcal{A}(I)  $, and similarly for $ \mathfrak{B}$. The projection $e=[\mathfrak{B} \Omega]$ satisfies $e(x \Omega) = \varepsilon(x) \Omega$ and $exe = \varepsilon(x)e$ for all $x$ in $\mathfrak{A}$. 
\end{prop}
\begin{proof}
	Denote by $\omega(\cdot) = (\Omega | \cdot \Omega)$ the vacuum  state and by $e$ the orthogonal projection onto  $\hil_\mathcal{B} = \overline{\bigvee_{I \in \mathcal{K}} \mathcal{B}(I) \Omega}$. As mentioned above, for each interval $I$ in $\mathcal{K}$ we have that $e(x \Omega) = \varepsilon_I(x) \Omega$ for  $x$ in $\mathcal{A}(I) $. As the vacuum is locally faithful by Reeh-Schlieder, this implies that the conditional expectations are compatible, namely $\varepsilon_J$ is an extension of $\varepsilon_I$ whenever $I \subseteq J$. Therefore, for $x$ in the union $\bigcup_{I \in \mathcal{K}} \mathcal{A}(I)$ we can define $\varepsilon(x)  $ by setting $\varepsilon(x) = \varepsilon_I(x)  $ whenever $x$ is in  $\mathcal{A}(I)$. The map $\varepsilon$ is bounded since every $\varepsilon_I  $ has unital norm, hence  we can continuously extend $\varepsilon$ to $\mathfrak{A}$ (this procedure is sometimes known as the BLT theorem). Finally, $\varepsilon$ is a conditional expectation since by  continuity $\varepsilon$ is a positive $ \mathfrak{B}$-linear projection, and the identity $\omega = \omega \cdot \varepsilon$ follows as well. To prove the last statement, one first shows that $x \Omega \mapsto \varepsilon(x) \Omega$, with $x$ in $\mathfrak{A}$, is a well defined projection onto $e \hil$. The identity $exe = \varepsilon(x)e$ follows \cite{jones1983index}. 
\end{proof}

\begin{lem} \label{lem:twist}
	Let $\mathcal{B} \subseteq \mathcal{A}$ be an inclusion  of  M\"obius covariant  nets with Jones projection $e=[\mathfrak{B} \Omega]$. Assume $(\mathcal{A}, U, \Omega)$ to be twisted-local with twist operator $Z_A$. If $eZ_A=Z_Ae$, then  $\mathcal{B}$  is twisted-local. 
\end{lem}
\begin{proof}
	We show $Z_B = eZ_Ae$ to be a twist operator for $\mathcal{B} $. Clearly $Z_B$ fixes the vacuum vector and commutes with $U$. Then, by Theorem \ref{thm:jones} and twisted locality we have  the chain of  inclusions $Z_B \mathcal{B}(I')Z_B^* =  Z_B \mathcal{A}(I')Z_B^* \subseteq e \mathcal{A}(I)' e \subseteq e \mathcal{B}(I)' e $. The thesis follows.
\end{proof}

\begin{defn}
	\cite{longo2021neumann} Let $(\mathcal{A}, U, \Omega)$ be a twisted-local net satisfying the split property. Given a couple of distant open intervals $I$ and $J$ of the circle, denote by $\mathcal{F} $ the canonical intermediate type I factor associated to the bipartite system $\mathcal{A}(I) \vee \mathcal{A}(J)$. We define the {\em canonical entanglement entropy} of $\omega$ with respect to $(I,J)$ the von Neumann entropy
	\[
	E_C(\omega) = S_\mathcal{F}(\omega) =  S_{\mathcal{F}'}(\omega) \,.
	\]
\end{defn}

\begin{remark} \label{rem:ec}
In the notation of the previous definition, assume $(\mathcal{A}, U, \Omega)$  to be local. Since $\omega$ is a pure state on $B(\hil) = \mathcal{F} \vee \mathcal{F}'$, then by the monotonicity of the relative entropy and \eqref{eq:mi} we find that 
\[
E_I(\omega) \leq 2 E_C(\omega) \,,
\]
with $E_I(\omega) $ the mutual information of the bipartite system $\mathcal{A}(I) \vee \mathcal{A}(J)$. 
\end{remark}

\begin{thm} \label{thm:conditional}
	Let $(\mathcal{A}, U, \Omega)$ be a twisted-local net on some Hilbert space $\hil$ with twist operator $Z_A$. Consider a twisted-local subnet  $\mathcal{B}$ of $\mathcal{A}$ satisfying the assumptions of Lemma \ref{lem:twist}. Assume also $\mathcal{A}$ and $\mathcal{B}$ to both  satisfy the split property. Denote by  $\mathcal{F}_A $ and $ \mathcal{F}_B$   the   canonical intermediate type I factors corresponding to the inclusions $\al(I) \subseteq   \al(\tilde{I})$  and $\mathcal{B}(I) \subseteq   \mathcal{B}(\tilde{I})$  respectively, with $\overline{I} \subseteq \tilde{I}$.  If $\varepsilon$ and $e$ are as in Proposition \ref{prop:conditional}, then $\varepsilon(\mathcal{F}_A)e = \mathcal{F}_B$. In particular, 
	\begin{equation} \label{eq:inequality}
		S_{\mathcal{F}_B}(\omega) \leq S_{\mathcal{F}_A}(\omega) \,.
	\end{equation}
\end{thm}
\begin{proof}
		We denote by $J_A$ the modular conjugation of $\mathcal{A}(I) \vee \mathcal{A}(\tilde{I})'$ and by $J_B$ the modular conjugation of $\mathcal{B}(I) \vee \mathcal{B}(\tilde{I})'$. By Proposition  \ref{prop:conditional} we have a conditional expectation $\varepsilon \colon \mathcal{A}(\tilde{I}) \to  \mathcal{B}(\tilde{I}) $, but by Lemma \ref{lem:twist} we also have that $e\mathcal{A}({I})'e = Z_A \mathcal{B}({I}') Z_A^*e $. Therefore, by Theorem \ref{thm:jones} the map $\varepsilon$ restricts to a conditional expectation $\varepsilon \colon \mathcal{A}({I})' \cap \mathcal{A}(\tilde{I}) \to  Z_A \mathcal{B}({I}') Z_A^* \cap \mathcal{B}(\tilde{I}) $, and thanks to the  considerations described in the proof of Corollary \ref{cor:natural}  we can claim that $eJ_A=J_Ae = J_Be$. We can now state that the conditional expectation $\varepsilon$ maps $  \mathcal{A}(\tilde{I}) \cap J_A \mathcal{A}(\tilde{I})J_A  $ onto $  \mathcal{B}(\tilde{I}) \cap J_A \mathcal{B}(\tilde{I})J_A  $:  clearly $\varepsilon(\mathcal{A}(\tilde{I})) = \mathcal{B}(\tilde{I})$, but thanks to Theorem \ref{thm:jones} and the  previous remark we also have  $\varepsilon(J_A \mathcal{A}(\tilde{I})J_A) = J_A\mathcal{B}(\tilde{I})J_A$. It follows that  $\varepsilon(\mathcal{F}_A)e = \mathcal{F}_B$. Finally, inequality \eqref{eq:inequality} is a consequence of   Theorem \ref{thm:conditional} and Lemma \ref{lem:inequality}.
\end{proof}

We think Theorem \ref{thm:conditional} to be one of the main results of this work. We point out that, even though we focused on twisted-local nets on the circle, Proposition \ref{prop:conditional} and Theorem \ref{thm:conditional} easily extend to any inclusion on twisted-local Haag-Kastler nets, provided that the Bisognano-Wichmann property is satisfied.



\begin{remark} \label{rem:law1}
	Let $(\mathcal{A}, U, \Omega)$ be a twisted-local net satisfying the split property. Given a couple of disjoint and distant open intervals $I$ and $J$ of $S^1 \setminus \{ -1\}$, by stereographic projection we can identify them with open intervals $\tilde{I}$ and $\tilde{J}$ of the real line \cite{panebianco2021loop}. Define $s = \text{dist}(\tilde{I}, \tilde{J})$ and denote by $E_C(s)$ the corresponding canonical entanglement entropy.  By monotonicity of the relative entropy, one can apply 	Theorem 17 of \cite{hollands2018entanglement} to provide a lower bound for  $E_C(s)$ in the limit $s \to 0$. Similarly, always by assuming dilation covariance, in higher dimension we can claim that the canonical entanglement entropy satisfies lower bounds of area law type.
\end{remark}

\begin{thm} \label{thm:main}
	Let $(\mathcal{A}, U, \Omega)$  be one of the following twisted-local conformal nets:
	
	1. the free Fermi net, 
	
	2. some $LSU(n)$-conformal net of level $\ell \geq 1$, 
	
	3. the $U(1)$-current, 
	
	4. the Virasoro net with  central charge  given by 
	\begin{equation} \label{eq:central}
		c = \frac{\ell (n^2 -1)}{\ell + n}  \,, 
	\end{equation}
with $\ell \geq 1$ and $n \geq 2$ integers. 

 If $ \mathcal{F} $  is  the  canonical intermediate type I factor given by the inclusion  $\mathcal{A}(I) \subseteq   \mathcal{A}(\tilde{I})$   with $\overline{I} \subseteq \tilde{I}$, then 
	\begin{equation}  \label{eq:finite}
		S_{\mathcal{F}}(\omega) < + \infty \,.
	\end{equation}
\end{thm}
\begin{proof}
	The finiteness property \eqref{eq:finite} has been proved on the free Fermi net  in \cite{longo2021neumann}, where explicit estimates can be found. However, by similarity of the first quantization Hilbert spaces, the same proof  can be replicated on  any $LSU(n)$-conformal net of level $\ell =1$  \cite{panebianco2021loop, wassermann1998operator}. The explicit estimates provided in \cite{longo2021neumann} still apply in this setting.   Therefore, since the embedding $LSU(n) \subseteq LSU(n \ell) $ gives rise to all the  $LSU(n)$-conformal nets of level $\ell \geq 1$  \cite{wassermann1998operator}, the estimates of \cite{longo2021neumann} apply to any $LSU(n)$-conformal net. 	Since  $LSU(n)$-conformal nets are local, we can apply Theorem \ref{thm:conditional} to any conformal subnet of these models like  the Virasoro net with central charge given by \eqref{eq:central} \cite{panebianco2021loop}. Theorem \ref{thm:conditional} cannot be applied to any conformal subnet of the free Fermi net, since in this case graded locality rather than locality holds \cite{longo2021neumann}. However, it can be easily checked that such a requirement holds at least for the    $U(1)$-current model \cite{longo2021neumann}. 
\end{proof}


Theorem \ref{thm:main}, which can be summarized as  a generalization of \cite{longo2021neumann}  by using Theorem \ref{thm:conditional},  is the main result of this work. We point out that the proof exhibited in  \cite{longo2021neumann} heavily depends on the structure of the free Fermi net. However, more  in general the finiteness property \eqref{eq:finite} is expected to rely on some nuclearity condition of the system such as the trace-class property \cite{hollands2018entanglement, otani2018toward}. To support this conjecture,  in the next section we provide a few results motivated by Remark~\ref{rem:ec}.

	\section{Modular nuclearity and entanglement} \label{sec:6}
	
	
	

In the previous section we proved the finiteness of some entanglement entropy to be finite on some twisted-local nets on the circle. Works on this topic suggest such an entanglement measure  to be  finite by assuming some modular nuclearity condition \cite{hollands2018entanglement, longo2021neumann}. Even though a general proof on a model independent ground is still lacking, in this section we provide a few  results in this direction.

	\begin{lem}\label{lem:HS-lem3}
		Let $(A, B)$ be a standard split pair with standard vector $\Omega$ inducing a state $\omega$. We  assume modular $p$-nuclearity to hold for some $0 < p \leq 1$, namely the $p$-partition function  \eqref{eq:nu} is finite for some $0 < p \leq 1$. Given $\epsilon >0$, there are sequences of normal linear functionals $\phi_j$ on $A$ and $\psi_j$ on $B$ such that 
		\begin{equation} \label{eq:sep}
			\omega(ab) = \sum_j \phi_j(a) \psi_j(b) \,, \quad a \in A \,, \; \;  b \in B \,,
		\end{equation}
		and $\sum_j \Vert \phi_j \Vert^p \Vert \psi_j \Vert^p < z_p^p + \epsilon $.
	\end{lem}
	
	\begin{proof}
	 We assume $z_p = \Vert \Xi_A \Vert_p$ and we follow Lemma 3 of \cite{hollands2018entanglement}.  Given $a$ in $A$ and $b$ in $B$, we note that 
		\begin{equation*}
			\begin{split}
				\omega(ab) & = (\Omega | ab \Omega) = ((\Delta^{1/4} + \Delta^{-1/4})^{-1}(1+\Delta^{-1/2})b^* \Omega | \Delta^{1/4} a \Omega) \\
				& = ((\Delta^{1/4} + \Delta^{-1/4})^{-1}(b^* + JbJ)\Omega | \Xi_A(a)) \,,
			\end{split}
		\end{equation*}
		where  $\Delta = \Delta_{B', \Omega}$ and $J=J_{B, \Omega}$. If $z_p$ is finite and $\epsilon >0$, then there are sequences of positive normal functionals $\phi_j$ on $A$ and vectors $\xi_j$ in $\hil$ such that
		\[
		\Xi_B(a) = \sum_j \phi(a) \xi_j \,, \quad a \in A \,,
		\]
		and $\sum_j \Vert \phi_j \Vert^p \Vert \xi_j \Vert^p < z_p^p + \epsilon$. Define now normal functionals $\psi_j$ on $B$ by
		\begin{center}
			$  			\psi_j(b) = ((\Delta^{1/4} + \Delta^{-1/4})^{-1}(b^* + JbJ)\Omega| \xi_j) \,,  $
		\end{center}
		and note that $\Vert \psi_j \Vert \leq \Vert \xi_j \Vert$ because of the estimate $\Vert (\Delta^{1/4} + \Delta^{-1/4})^{-1} \Vert \leq 1/2$ and the spectral calculus. Putting both paragraphs together we find the conclusion.
	\end{proof}

	Before providing a corollary of the previous lemma, we describe a general procedure known as {\em polarization} of a functional. Let $\omega$ be a continuous functional.  We will say that $\omega$ is {\em self-adjoint} if $\omega = \omega^*$, with $\omega^*(x) = \overline{\omega(x^*)}$ the {\em conjugate} of $\omega$.  By use of   $\omega^*$ one can write $\omega = \phi + i \psi$, with $\phi$ and $\psi$ self-adjoint. Then, after applying a Jordan decomposition on both $\phi$ and $\psi$, we can write $\omega= \sum_{k=0}^3 (i)^\alpha \omega_\alpha$, with $\omega_\alpha$ positive. The inequality $\Vert \omega_\alpha \Vert \leq \Vert \omega \Vert$ can also be proved, and $\omega_\alpha$ are all normal if $\omega$ is.

\begin{cor} \label{cor:four}
	With the notation of the previous lemma, for every $\epsilon > 0$ we can write $\omega = (1+\lambda) \omega_+ - \lambda \omega_-$, where $\omega_\pm$ are separable states and $(1+\lambda)^p \leq 4 (z_p^p + \epsilon)$. 
\end{cor}

\begin{proof}
	By polarization, we can decompose $\phi_j = \sum_{\alpha=0}^3 (i)^\alpha \phi_j^\alpha$ and  $\psi_j = \sum_{\alpha=0}^3 (i)^\alpha  \psi_j^\alpha$ in four positive normal functionals. One also has that  $ \Vert \phi_j^\alpha \Vert \leq \phi_j$ holds, and similarly for $\psi_j$. Since $\omega$ is positive, then after the identification $A \vee B \cong A \otimes B$ we find 
	\[
	\omega = \sum_j \sum_{\alpha=0}^{3} \phi_j^\alpha \otimes \psi_j^{4 - \alpha}- \sum_j \sum_{\alpha=0}^{3} \phi_j^\alpha  \otimes \psi_j^{2 - \alpha} \,,
	\]
	namely $\omega$ is difference of two separable functionals. The thesis follows.
\end{proof}

	\begin{lem}\label{lem:HS-lem4}
		With the hypotheses of  Lemma \ref{lem:HS-lem3}, assume $\omega$  to have an expression like in \eqref{eq:sep} and assume $\mu_p = \sum_j \Vert \phi_j \Vert^p \Vert \psi_j \Vert^p $ to be finite for some $0 < p \leq 1$. Then there is a separable positive linear functional $\sigma$ such that $\sigma \geq \omega$ on $A \vee B$ and $\Vert \sigma \Vert^p = \mu_1^p \leq \mu_p $. 
	\end{lem}
	
	\begin{proof}
		We follow \cite[Lemma 4]{hollands2018entanglement}. By polar decomposition there are partial isometries $u_j$ in $A$ such that $\phi(u_j \cdot ) \geq 0$ on $A$  and $\phi_j(u_j u_j^* \cdot ) = \phi_j$. It follows in particular that $\phi_j (u_j) = \Vert \phi_j \Vert$ and 
		\begin{center}
			$  		\bar{\phi}_j(a) = \overline{\phi_j(u_j u_j^* a^*)} = \phi_j (u_j (u_j^*a^*)^*) = \phi_j(u_j a u_j)  $
		\end{center}
		for all $a$ in $A$, where we used the fact that $\phi_j(u_j \cdot )$ is hermitian (here $\bar{\psi}(a) = \overline{\psi(a^*)}$). Similarly, there are partial isometries   $v_j$ in $B$ such that $\psi_j (v_j \cdot ) \geq 0$ and $\psi_j (v_j v_j^* \cdot) = \psi_j$. Note that the positive linear functional $\rho_j = \phi_j (u_j \cdot ) \otimes \psi_j (v_j \cdot)$ is separable. Writing $w_j = u_j \otimes v_j$ we then define
		\[
		\sigma_j (\cdot) = \frac{1}{2}\rho_j(\cdot) + \frac{1}{2}\rho_j( w^* \cdot w) \,,
		\]
		which is also separable, because $w$ is a simple tensor product. Furthermore, 
		\begin{center}
			$  	\Vert \sigma_j \Vert = \sigma_j(1) = \rho_j(1) = \Vert \phi_j \Vert  \Vert \psi_j \Vert\,,   $
		\end{center}
		and also
		\[
		0 \leq \frac{1}{2} \rho_j((1-w^*) \cdot (1-w)) = \sigma_j - \frac{1}{2} (\phi_j \otimes \psi_j + \bar{\phi}_j \otimes \bar{\psi}_j ) \,.
		\]
		We conclude that  $\sigma = \sum_j \sigma_j$ is a separable positive linear functional with 
		\[
		\sigma \geq \frac{1}{2} \sum_j  (\phi_j \otimes \psi_j + \bar{\phi}_j \otimes \bar{\psi}_j ) = \frac{1}{2} ({\omega} + \overline{\omega}) = \omega \,. 
		\]
		and $ \Vert \sigma \Vert^p =  \big( \sum_j \Vert \sigma_j \Vert \big)^p \leq  \sum_j \Vert \sigma_j \Vert ^p =  \mu_p $. 
	\end{proof}
	
	\begin{remark}
		Notice that, by the two previous lemmas, we have $z_p \geq 1$. 
	\end{remark}

	\begin{thm} \label{thm:mutual}
		Let $(A, B)$ be a standard split pair of hyperfinite factors. Assume the $p$-partition function \eqref{eq:nu} to be finite for some $0 < p <  1$. Then the mutual information is finite, with
		\begin{equation} \label{eq:mt}
			E_I(\omega) \leq c_p z_p + \eta(z_p-1) - \eta(z_p)  \,,
		\end{equation}
		where  $c_p = \frac{1}{(1-p)e}$ and $\eta(t) = -t \ln t$.
	\end{thm}

	\begin{proof}
		We begin the proof with a general remark. Consider a state $\omega$ on $A \otimes B$, with $A$ and $B$ finite dimensional type I factors. If   $\omega = \sum_j \lambda_j \omega_j$ is  a convex decomposition in states, then by \eqref{eq:mi} we have
		\begin{equation} \label{eq:concave}
		E_I (\omega) \geq \sum_j \lambda_j E_I(\omega_j) - \sum_j \eta(\lambda_j) \,.
		\end{equation}
		By monotonicity of the relative entropy, the same expression holds if $A$ and $B$ are hyperfinite factors. Therefore, by  the previous lemmas  for every $\epsilon >0$ we have a separable functional $\sigma \geq \omega$ such that $\Vert \sigma \Vert ^p \leq z_p^p + \epsilon $. By setting $\hat{\sigma} = \sigma / \Vert \sigma \Vert $  we can write $\omega = \Vert \sigma\Vert \hat{\sigma} - \Vert \tau\Vert \hat{\tau}$ and apply \eqref{eq:concave} to notice that
		\[
		\Vert \sigma \Vert E_I(\hat{\sigma}) \geq E_I(\omega) + \eta (\Vert \sigma\Vert) - \eta (\Vert \sigma\Vert - 1) \,,
		\] 
		where we used the positivity  property $E_I(\hat{\tau}) \geq 0$. We then recall that the separability of $\hat{\sigma}$ implies that $E_I(\hat{\sigma}) \leq  \Vert \sigma \Vert ^{-p} c_p (z_p^p + \epsilon) $. Therefore, the claimed estimate follows from the inequality $\eta(t) \leq c_p t^p$ for $p<1$  and  the monotonicity of $\eta(s-1) - \eta(s)$ for $s \geq 1$. 
	\end{proof}

	\begin{remark}
		Due to the  inequality $S(\varphi \Vert \omega) \geq \Vert \varphi - \omega \Vert^2/2$ \cite{ohya2004quantum}, we can use the previous result to estimate the distance between the states $\omega$ and $\omega \otimes \omega$. See \cite{morinelli2018conformal} for related issues concerning the split property.
	\end{remark}

	This upper bound of the mutual information, which by monotonicity becomes sharper for smaller $p$,  has been inspired by the works  \cite{hollands2018entanglement, narnhofer2002entanglement}. Lower bounds on the mutual information follow from Theorem 17 and  Theorem 18 of \cite{hollands2018entanglement}. We recall that modular $p$-nuclearity holds on scalar free fields for any $p>0$ \cite{lechner2016modular} and  on conformal nets satisfying the trace class property \cite{buchholz2007nuclearity}. As we will later describe, it also holds on a wide family of $1+1$-dimensional integrable models with factorizing S-matrices \cite{alazzawi2017inverse, lechner2008construction, lechner2018approximation}. All these models also satisfy the hyperfiniteness of the local algebras, hence Theorem \ref{thm:mutual} can be applied in these settings. We now follow \cite{otani2018toward} and we show the finiteness of some “tailored” entanglement entropy under the assumption of modular $p$-nuclearity for some $0  < p < 1$. 
	
	\begin{defn} \label{def:ip} Let $(A, B)$ be a standard split pair of von Neumann algebras on a Hilbert space $\hil$ with standard vector $\Omega$. If  $u \colon \hil \to \hil \otimes \hil$ is a unitary implementing the natural  isomorphism $A \vee B \cong A \otimes B$, then we will denote by $R_u = u^{-1}(B(\hil) \otimes \mathbbm{1} )u$ the corresponding  type I factor. We define  an {\em intermediate pair} any such pair $(u, R_u)$.
	\end{defn}

	\begin{defn} 
		Let $(A, B)$ be a standard split pair as in the previous definition. Given a state $\psi$ on $B(\hil)$, we  call the {\em intermediate entanglement entropy} of $\psi$ the functional
		\[
		I(\psi)= \sup_{(u, R_u)} \inf_{\phi, \lambda} \frac{1}{\lambda} S(\phi)\,,
		\]
		where supremum is over all intermediate pairs, the  infimum is over all states $\phi$ on $B(\hil)$ and real numbers $0 < \lambda \leq 1$ such that $\phi \geq \lambda \psi$ on $A \vee B$ and $S(\phi)$ is the von Neumann entropy of $\phi$ on the intermediate type I factor $R_u$.
	\end{defn}
	
	
	\begin{thm} \label{thm:otani} 	Let $(A, B)$ be a standard split pair with  standard vector $\Omega$ inducing a state $\omega$. Denote by $z_p$ the $p$-partition function \eqref{eq:nu}.  If $z_p$ is finite for some $0 < p < 1$, then the intermediate entanglement entropy is finite. Explicitly, 
		\begin{equation}\label{eq:otani}
			I(\omega)\leq  z_p \ln z_p +  c_p z_p^p  \, .
		\end{equation}
	\end{thm} 
	
	\begin{proof}
		The proof consists of a computation that does not depend on the choice of the intermediate pair, which is therefore implicit in what follows. Through the natural isomorphism $A \vee B \cong A \otimes B$ we will  identify $A$ with $A \otimes \mathbbm{1}$ and $B$ with $\mathbbm{1}\otimes B $. Lemma \ref{lem:HS-lem4} gives a separable dominating normal functional $\sigma  \geq \omega$ with $\Vert \sigma \Vert^p \leq z_p^p + \epsilon$ for $\epsilon >0$ arbitrarily small.  We utilize the separability of $\hat{\sigma} = \sigma / \Vert \sigma \Vert $ over the bipartite system $  A \otimes B $ and decompose it into positive, normal functionals, say  $\hat{\sigma}=\sum_j \phi_j\otimes\psi_j$. Without loss of generality we can assume $\phi_j$ to be states on $A$. Now we notice that $\phi_j \otimes \psi_j $ is a normal positive functional on $A \otimes B$, hence it can be extended by taking a representative vectors. Since such extension has same norm, we can extend $\sigma$ to a separable positive functional on $B(\hil) \otimes B(\hil)$ in such a way that still $\Vert \sigma \Vert^p \leq z_p^p + \epsilon$.  We  introduce some further notation by setting  $\eta(t) = - t \ln t $ and  $1/c_p = {(1-p)e}$. Therefore, we have
		\begin{equation*} \label{eq:est1}
			\Vert \sigma  \Vert S(\hat{\sigma})  \leq \Vert \sigma  \Vert \ln \Vert \sigma  \Vert + \sum_j \eta (\Vert \psi_j \Vert) \leq \Vert \sigma  \Vert \ln \Vert \sigma  \Vert + c_p \sum_j \Vert \psi_j \Vert^p \,,
		\end{equation*}
		and the claimed estimate follows from the arbitrarity of $\sigma$. 
	\end{proof}

	\section{Application to  integrable models}  \label{sec:7}

	\begin{defn}
		A local quantum field theory $(\mathcal{A}, U, \Omega)$ on the Minkowski space is said to satisfy the {\em split property (for double cones)} if the inclusion $\mathcal{A}(\mathcal{O}_1) \subseteq \mathcal{A}(\mathcal{O}_2) $ is split whenever ${\mathcal{O}_1} \subset \mathcal{O}_2$ is an inclusion of  double cones such that $\overline{\mathcal{O}_1} \subset \mathcal{O}_2$ .
	\end{defn}

	The split property ensures some statistical independence property of the considered model. The split property does not hold for unbouded regions like wedges in more than two spacetime dimensions \cite{buchholz1974product, lechner2008construction}.  In the literature there are several criteria which are known to imply the split property, where many of them are  referred to  as “nuclearity conditions”. In order to formulate this condition in a local Haag-Kastler net on a $d$-dimensional Minkowski space, one considers a region $\mathcal{O} \subseteq \R^d$ and a parameter $\beta > 0$ representing the inverse temperature. One defines 
	\begin{equation} \label{eq:enc}
		\Theta_{\beta, \mathcal{O}} \colon \mathcal{A}(\mathcal{O}) \to \hil \,, \quad 	\Theta_{\beta, \mathcal{O}} (A) = e^{- \beta H} A \Omega \,,
	\end{equation}
	where $H = P_0$ denotes the Hamiltonian with respect to the time direction $x_0$. 
	\begin{defn}
		A local quantum field theory on the Minkowski space $\R^{d}$ is said to satisfy the  {\em energy nuclearity condition} if the maps \eqref{eq:enc} are nuclear for any bounded region $\mathcal{O}$ and any inverse temperature $\beta >0$. Moreover, there must exist constants $\beta_0, n>0$ depending on $\mathcal{O}$ such that the nuclear norm of  $		\Theta_{\beta, \mathcal{O}} $ is bounded by
		\[
		\Vert \Theta_{\beta, \mathcal{O}}  \Vert_1 \leq e^{(\beta_0 / \beta)^n} \,, \quad \beta  \to 0  \,.
		\]
	\end{defn}

	\begin{defn}
		A local quantum field theory on the Minkowski space is said to satisfy the {\em modular nuclearity condition (for double cones)} if the inclusion  $\mathcal{A}(\mathcal{O}_1) \subseteq \mathcal{A}(\mathcal{O}_2) $ is  modular nuclear  whenever ${\mathcal{O}_1} \subset \mathcal{O}_2$ is an inclusion of double cones such that $\overline{\mathcal{O}_1} \subset \mathcal{O}_2$.
	\end{defn}
	
	A modular nuclear  inclusion of factors is split, and a split  inclusion of factors implies the compactness of the map \eqref{eq:xi} \cite{buchholz1990nuclear}. As mentioned in  \autoref{sec:4}, the previous nuclearity conditions can be changed into {\em modular $p$-nuclearity conditions} by requiring to the considered maps to be of type $l^p$ for some $0 < p \leq 1$.  Modular $p$-nuclearity has been proved in the theory of a  scalar free field for any $p>0$ \cite{lechner2016modular} and  holds on conformal nets satisfying the trace class property \cite{buchholz2007nuclearity}. \\
	
			A class of integrable models on $\R^{1+1}$ that satisfies modular $p$-nuclearity for wedges was constructed in \cite{lechner2008construction}. The only input needed are the mass $m$ and a $2$-body scattering matrix. We review the structure and some properties of these models. The statements and the construction in greater detail can be found in \cite{lechner2008construction}.

	
	\begin{defn}
		A {\em (2-body) scattering function} is an analytic function $S_2 \colon S(0, \pi) \to \C$ which is bounded  and continuous on the closure of this strip and satisfies the equations
		\[
		\overline{S_2(\theta)} = S_2(\theta)^{-1} = S_2(- \theta) = S_2(\theta + i \pi) \,, \quad \theta \in \R \,.
		\]
		The set of all the scattering functions will be denoted by $\mathcal{S}$. For $S_2$ in $\mathcal{S}$, we define
		\[
		\kappa(S_2) = \inf  \{ \text{Im} \, \zeta \colon \zeta \in S(0, \pi/2) \,, \quad S_2(\zeta) =0\} \,.
		\]
		The subfamily $\mathcal{S}_0 \subset \mathcal{S}$ consists of those scattering functions $S_2$ with $\kappa(S_2)>0$ and for which
		\[
		\Vert S_2 \Vert_\kappa  = \sup  \big\{ |S_2(\zeta)| \colon \zeta \in \overline{S(- \kappa, \pi + \kappa)} \big\} < + \infty \,, \quad \kappa \in (0, \kappa(S_2)) \,.
		\]
		The families of scattering functions $\mathcal{S}$ and $ \mathcal{S}_0$ can then be divided into “bosonic” and “fermionic” classes according to 
		\begin{equation*}
			\begin{split}
				\mathcal{S}^\pm & = \big\{ S_2 \in \mathcal{S} \colon S_2(0) = \pm 1\big\} \,, \quad \mathcal{S} = \mathcal{S}^+ \cup \mathcal{S}^- \,, \\
				\mathcal{S}_0^\pm & = \big\{ S_2 \in \mathcal{S}_0 \colon S_2(0) = \pm 1\big\} \,, \quad \mathcal{S} = \mathcal{S}_0^+ \cup \mathcal{S}_0^- \,.
			\end{split}
		\end{equation*}
	\end{defn}

		In general, the $S$-matrix of a QFT model with interaction is difficult to handle and has been subject of research (see e.g. \cite{LecPhD} and references therein). In particular, there is few knowledge about the higher $S$-matrix elements $S_{n,m}$ with $n,m >2$. In $1+1$ dimensions, there exist $S$-matrices, called {\em factorising $S$-matrices}, that are completely determined by the two-particle $S$-matrix. The entries $S_{n,m}$ vanish for $n \neq m$ and $S_n = S_{n,n}$ is the product of two-particle $S$-matrices. 	We now paraphrase the main steps of constructing a net of von Neumann algebras on $\R^{1+1}$ from a ($2$-body) scattering function $S_2$ \cite{lechner2008construction}. \\
	
		One can parameterize the mass-shell $\Omega_m$ with the rapidity $p(\theta) = m( \cosh \theta,\sinh \theta)$, with  $\theta \in \R$, so that the one-particle space simplifies to $\hil_1 = L^2(\R, d \theta)$. The next step is to define a “$S_2$-symmetric” Fock-space over $\hil_1$. For that, we define a unitary representation $D_n$ of the symmetric group $\mathfrak{S}_n$ on $\hil_1^{\otimes^n}$ that takes the $S$-matrix $S_2$ into account. The complete $S_2$-symmetric projection is defined as $E_n = (1/n!) \sum_{\sigma \in \mathfrak{S}_n} D_n(\sigma)$. The “$S_2$-symmetric” $n$-particle space is $\hil_n = E_n \hil_1^{\otimes^n}$ and the “$S_2$-symmetric” Fock space is the direct sum of the $n$-particles spaces: $\hil= \C \Omega \bigoplus_{n \geq 1} \hil_n $. For $\chi \in \hil$, we define the {\em annihilation} $z(\chi)$ and {\em creation operator} $z^\dagger(\chi)$ with the $S_2$-symmetric projection analogous to the usual creation and annihilation operators. The creation and annihilation operators are closable and have a common core and are mutual adjoints of each other. For a Schwartz function $f \in \mathcal{S}(\R^{1+1})$,	one-particle vectors $f^\pm$ are defined as the (anti) Fourier transformation of $f$ restricted to the mass shell. Furthermore, we define an involution $J$ by $(J\Psi)_n(\theta_1,\dots , \theta_n) = \Psi_n(\theta_1,\dots , \theta_n)$. With $f^*(x) = \overline{f(-x)}$, the (closable and essentially self-adjoint for real $f$) {\em field operators} $\phi(f) $ and $\phi'(f)$ are defined as 
	\begin{equation*}
		\phi(f)  = z^\dagger(f^+) + z(f^-) \,, \quad \phi'(f)  = J\phi(f^*) J \,.
	\end{equation*}
If $f$ and $g$ belong to $ \mathcal{S}(\R^{1+1})$, with $f$ supported in $W_R$ and $g$  supported in $W_L$, then $[\phi'(f), \phi(g)]\Psi = 0 $ for any $\Psi$ belonging to a dense common core. Finally, this construction gives to a local net $W \mapsto \mathcal{A}(W)$ on  wedges defined by 
\begin{equation*}
	\begin{split}
	\mathcal{A}(W_L + x) & = \big\{ e^{i \phi(f) }\colon f \in \mathcal{S}(W_L+x,\R) \big\}'' \,, \\
\mathcal{A}(W_R + x) & = \big\{ e^{i \phi'(f) }\colon f \in \mathcal{S}(W_R+x,\R) \big\}'' \,.
	\end{split}
\end{equation*}
The algebra of observables localized in a double cone $\mathcal{O} = W_1 \cap W_2$ is defined as
\[
	\mathcal{A}(W_1 \cap W_2)  = \mathcal{A}(W_1) \cap \mathcal{A}(W_2) \,,
\]
	and for arbitrary open regions $\mathcal{Q} \subseteq \R^{1+1}$ we put $\mathcal{A}(\mathcal{Q})$ as the von Neumann algebra generated by all the local algebras $\mathcal{A}(\mathcal{O})$ with $\mathcal{O} \subseteq \mathcal{Q}$.  In \cite{lechner2008construction}, the author shows that the above models with scattering function $S_2$ in $\mathcal{S}_0^-$ satisfy the axioms of AQFT. Moreover, the inclusion $\mathcal{A}(W_R+{ \bf s})\subset \mathcal{A}(W_R)$ is modular $p$-nuclear for large enough $s$, where ${\bf s}=(0,s)$.

	\begin{thm} \label{thm:gandalf} \cite{alazzawi2017inverse, lechner2008construction, lechner2018approximation} Let $(\mathcal{A}, U, \Omega)$ be an integrable quantum field theory on $\R^2$ with factorizing $S$-matrix $S_2 \in \mathcal{S}$. Define
		\begin{equation} \label{eq:nuclear}
			\Xi (s) \colon \mathcal{A}(W_R)  \to \hil \,, \quad \Xi(s)A = \Delta^{1/4} U({\bf s}) A \Omega \,, \quad s>0 \,,
		\end{equation}
		where $U({\bf s})$ is the unitary associated to the translation of ${\bf s} = (0,s)$ and $\Delta$ is the modular operator of $(\mathcal{A}(W_R), \Omega)$. If  $S_2 \in \mathcal{S}_0^-$, then there exist some finite  splitting distance $s_{\text{min}} < \infty$ such that $\Xi(s) $ is $p$-nuclear for all $p>0$ and $s > s_{\text{min}} $. 
	\end{thm}

	
	\begin{lem}
	The map \eqref{eq:nuclear} satisfies  $\Vert \Xi(s) \Vert_p \to 1$ as $s \to + \infty$.
	\end{lem}
\begin{proof}
	This fact can be noticed by reading the thesis \cite{ lechner2008construction}   and the subsequent works \cite{alazzawi2017inverse,  lechner2018approximation}.
\end{proof}

	\begin{remark}
		By studying carefully  \cite{alazzawi2017inverse}, \cite{lechner2008construction} and \cite{lechner2018approximation}, it should also follow  that the maps
\[
\Sigma (s) \colon \mathcal{A}(W_R)_{\text{sa}} \Omega \to \hil \,, \quad \Sigma(s)A\Omega = \Delta^{1/4} U({\bf s}) A \Omega \,, \quad s>0 \,,
\]
are $p$-nuclear for all $p>0$ and $s >s_{\text{min}}  $.
	\end{remark}


			We now compare the asymptotic behaviour of two different entanglement measures in the setting of these integrable models.  If we denote by $E_R(s)$ the vacuum relative entropy of entanglement corresponding to the wedge inclusion mentioned in Theorem \ref{thm:gandalf}, then by Lemma~\ref{lem:HS-lem3} and Lemma~\ref{lem:HS-lem4} it is easy to notice that \cite{hollands2018entanglement}
		\[
		E_R(s) \leq \ln \Vert  \Xi(s) \Vert_1 \to 0 \,, \quad s \to + \infty \,.
		\]
		However, if we denote by $E_I(s)$ the associated mutual information and we  apply the estimate \eqref{eq:mt} for any $0 < p < 1$ then we can state at most that
		\[
		\limsup_{s \to + \infty} E_I(s) \leq 1/e \,.
		\]
		Analogously, we can apply \eqref{eq:otani} to estimate the intermediate entanglement entropy in the limit $s \to + \infty$. With the same arguments and similar notation we find
				\[
		\limsup_{s \to + \infty} I(s) \leq 1/e \,.
		\]

	
	

	\section{Conclusions} \label{sec:8}
	
	We close this work with additional remarks that might be useful for future research in this area. In particular, we list a few conditions which are equivalent to the finiteness of the canonical entanglement entropy. \\

	The techniques of  \autoref{sec:5} rely on the presence of a  separable state $\sigma$ on the bipartite system $A\otimes B $ that dominates $\omega$. If $F$ is an intermediate type I factor $A\subseteq F \subseteq B'$ arising from the natural isomorphism $A \vee B \cong A \otimes B $ as in Definition \ref{def:ip}, then it  is possible to construct a separable functional on $  B(\hil) \otimes B(\hil) \cong F \vee F' $  that dominates $\omega$ on $F $ and on $F'$ by use of generalized conditional expectations \cite{accardi1982conditional}. \\

	More specifically, let $(A,B)$ be a standard split pair with finite $p$-partition function for some $0<p < 1$. As discussed in  \cite{accardi1982conditional, ohya2004quantum}, one has two $\omega$-preserving cpu maps, say $\varepsilon$  and $\varepsilon'$, induced by the inclusions  $A\subseteq F$ and $B \subseteq  F'$ respectively.  By the isomorphism $B(\hil) \otimes B(\hil) \cong  F \vee F' = B(\hil) $ we can then define a map $\varepsilon\otimes \varepsilon'  $ on $ B(\hil) $ extending both $\varepsilon $ and $ \varepsilon' $. If $\sigma$ is the dominating separable functional from Lemma \ref{lem:HS-lem4}, then ${\sigma}_0=\sigma \cdot (\varepsilon\otimes \varepsilon')$ dominates $ {\omega}_0= \omega \cdot  (\varepsilon\otimes \varepsilon')$. Notice that $\omega = {\omega_0} $ on $F$ and on $F'$, but in general not on $B(\hil)$. The functional ${\sigma_0} = \sum_j \phi_j\cdot\varepsilon \otimes  \psi_j\cdot\varepsilon'$   is separable with $\sum_j \lVert\phi_j\cdot\varepsilon\rVert^p \lVert\psi_j\cdot\varepsilon'\rVert^p=\mu_p  $ finite  (cf. Lemma \ref{lem:HS-lem4} for notation), and	in the notation of Theorem~\ref{thm:mutual} we have 
	\[
	E_I({\omega}_0) = S({\omega_0}\Vert \omega_F \otimes \omega_{F'}) \leq c_p z_p + \eta(z_p-1) - \eta(z_p)  \,,
	\]
	where the r.h.s. is finite by assumption. Unfortunately, this does not imply the finiteness of the canonical entanglement entropy since ${\omega_0}$ is not a pure state on $B(\hil)$. But we can make use of generalized conditional expectations to give an equivalent description of the canonical entanglement entropy. In particular, by use of  equation \eqref{eq:mi} and Lemma \ref{lem:entropy}   we can claim  that 
	\[
	2 E_C(\omega) = S_{B(\hil)}(\omega \Vert \omega_F \otimes \omega_{F'})=2 S_{B(\hil)}({\omega_0} )   = 2 H_{{\omega_0}} (F)  \,,
	\]
	with $H_{{\omega}_0} (F) = H_{{\omega}_0}^{B(\hil)} (F) $ the conditional entropy. The authors of \cite{dutta2021canonical} argued on grounds of physical arguments that
	\[
	E_C(\omega) \approx E_I^{F \vee B'}(\omega) = S(\omega \Vert \omega_F \otimes \omega_{B'}) \,,
	\] 
	and indeed it is reasonable to expect that the results of this work can be properly strengthened. For example, Theorem \ref{thm:mutual} implies that $E_I^{F \vee B'}(\omega)$ is finite if the $\omega$-preserving generalized conditional expectation from $F \vee  B' $ onto $ A\vee B'$ is a separable operation in the terminology of  Definition \ref{def:sep}. Another  strategy could be that of estimating the entanglement entropy of some energy cutoff of the vacuum state like  in \cite{otani2018toward} and then to operate some limit procedure.  A different approach is the one of \cite{longo2021neumann}, in which the authors use the language of standard subspaces. Unfortunately, even if completely rigorous, this last work heavily depends on the structure of the free Fermi nets, and a generalization of it seems quite challenging up to now. In the context of conformal nets, the authors expect the trace-class property to be a good assumption to start with \cite{longo2020lectures, otani2018toward}.

	\section*{Acknowledgements} 
	We thank Roberto Longo for suggesting us the problem. We also thank Yoh Tanimoto and Gandalf Lechner for explanations concerning the last section and Yoh Tanimoto, Gandalf Lechner and Ko Sanders for comments on an earlier version.
	
	BW is part of the INdAM Doctoral programme in Mathematics and/or Applications cofunded by  Marie  Sklodowska-Curie actions (INdAM-DP-COFUND 2015) and has received funding from the European Union's 2020 research and innovation programme under the Marie Sklodowska-Curie grant agreement No 713485. On behalf of all authors, the corresponding author states that there is no conflict of interest.

\end{document}